\documentclass[11pt,a4paper,oneside]{article}

\usepackage[english]{babel}
\usepackage{url}
\usepackage{booktabs}
\usepackage{amsmath,amsfonts,amsthm,amssymb}
\usepackage{latexsym,enumitem,xargs}
\usepackage{xcolor,bbm}
\usepackage{appendix}
\usepackage{authblk}

\usepackage{graphicx}

\usepackage[utf8]{inputenc}

\usepackage{natbib}
\usepackage[pagebackref=false,colorlinks=true,citecolor=black!70!blue,linkcolor=black!70!green,plainpages=false,pdfpagelabels=true]{hyperref}
\usepackage[all]{hypcap}

\renewcommand{\leq}{\leqslant}
\renewcommand{\geq}{\geqslant}

\renewcommand{\epsilon}{\varepsilon}
\renewcommand{\hat}{\widehat}

\renewcommand{\tilde}{\widetilde}

\def\1{\mathbbm{1}}
\def\R{\mathbb{R}}

\def\C{\mathbb{C}}
\def\Z{\mathbb{Z}}

\def\mci{\mathcal{I}}
\def\rme{\mathrm e}
\def\rmi{\mathrm i}
\def\rmd{\mathrm d}

\newcommand{\sign}{\mbox{sign}}

\newcommand{\atan}{\mbox{$\mathrm{arctan}$}}

\theoremstyle{plain}
\newtheorem{thm}{Theorem}
\newtheorem{lem}[thm]{Lemma}
\newtheorem{prop}[thm]{Proposition}
\newtheorem{cor}[thm]{Corollary}
\theoremstyle{definition}

\newtheorem*{defn}{Definition}

\theoremstyle{remark}
\newtheorem{rem}{Remark}

\title{New results on approximate {H}ilbert pairs of wavelet filters with
  common factors}
  
\author[a]{Sophie Achard}
\author[b]{Marianne Clausel}
\author[c]{Ir\`ene Gannaz}
\author[d]{Fran\c cois Roueff}

\affil[a]{Univ. Grenoble Alpes, CNRS, Grenoble INP, GIPSA-lab, 38000 Grenoble, France}
\affil[b]{Univ. Grenoble Alpes, CNRS, Grenoble INP, LJK, 38000 Grenoble, France}
\affil[c]{Universit\'e de Lyon,
CNRS UMR 5208, INSA de Lyon, Institut Camille Jordan, France}
\affil[d]{LTCI, T\'el\'ecom-Paristech, Universit\'e Paris-Saclay, France}

\date{\today}

\begin{document}

\maketitle

\setlength{\parindent}{0pt}
\setlength{\parskip}{\baselineskip}

\begin{abstract}
  In this paper, we consider the design of wavelet filters based on the
  \emph{Thiran} \emph{common-factor} approach proposed in
  \cite{Selesnick2001}. This approach aims at building finite impulse response
  filters of a \emph{Hilbert-pair} of wavelets serving as real and imaginary
  part of a complex wavelet.  Unfortunately it is not possible to construct
  wavelets which are both finitely supported and analytic.  The wavelet filters
  constructed using the \emph{common-factor} approach are then approximately
  analytic.  Thus, it is of interest to control their analyticity.  The purpose
  of this paper is to first provide precise and explicit expressions as well as
  easily exploitable bounds for quantifying the analytic approximation of this
  complex wavelet. Then, we prove the existence of such filters enjoying the
  classical perfect reconstruction conditions, with arbitrarily many vanishing
  moments.
\end{abstract}

\textbf{Keywords.} Complex wavelet, Hilbert-pair, orthonormal filter banks, common-factor wavelets
\section{Introduction}

Wavelet transforms provide efficient representations for a wide class of
signals. In particular signals with singularities may have a sparser
representation compared to the representation in Fourier basis. Yet, an
advantage of Fourier transform is its analyticity, which enables to exploit
both the magnitude and the phase in signal analysis. In order to combine both advantages of Fourier
and real wavelet transform, one possibility is to use a complex wavelet
transform. The analyticity can be obtained by choosing properly the wavelet filters. This may offer a
true enhancement of real wavelet transform for example in singularity extraction purposes. We refer to \cite{Selesnick-review,tay2007designing}
and references therein for an overview of the motivations for analytic wavelet
transforms.  A wide range of applications can be addressed using such wavelets
as image analysis \citep{chaux2006image}, signal processing
\citep{wang2010enhancement}, molecular biology
\citep{murugesan2015application}, neuroscience \citep{whitcher2005time}.

Several approaches have been proposed to design a pair of wavelet filters where
one wavelet is (approximately) the Hilbert transform of the other. Using this
pair as real and imaginary part of a complex wavelet allows the design of
(approximately) analytic wavelets.  The simplest complex analytic wavelets are
the generalized Morse wavelets, which are used in continuous wavelet transforms
in \cite{lilly2010analytic}. The approximately analytic Morlet wavelets can
also be used for the same purpose, see \cite{Selesnick-review}. However, for
practical or theoretical reasons, it is interesting to use discrete wavelet
transforms with finite filters, in which case it is not possible to design a
perfectly analytic wavelets. In addition to the finite support property, one
often requires the wavelet to enjoy sufficiently many vanishing moments,
perfect reconstruction, and smoothness properties. Among others linear-phase
biorthogonal filters were proposed in \cite{dualtree1,dualtree2} or q-shift
filters in \cite{qshift}. We will focus here on the \emph{common-factor}
approach, developed in \cite{Selesnick2001,Selesnick-thiran}. In
\cite{Selesnick-thiran} a numerical algorithm is proposed to compute the FIR
filters associated to an approximate Hilbert pair of orthogonal wavelet
bases. Improvements of this method have been proposed recently in
\cite{tay2010new,murugesan2014new}. The approach of \cite{Selesnick2001} is
particularly attractive as it builds upon the usual orthogonal wavelet base
construction by solving a Bezout polynomial equation. Nevertheless, to the best
of our knowledge, the validity of this specific construction have not been
proved. Moreover the quality of the analytic approximation have not been
thoroughly assessed. The main goal of this paper is to fill these gaps.  We
also provide a short simulation study to numerically evaluate the quality of
analyticity approximation for specific \emph{common-factor} wavelets.

After recalling the definition of Hilbert pair wavelet filters, the
construction of the Thiran's \emph{common-factor} wavelets following
\cite{Thiran,Selesnick-thiran} is summarized in Section~\ref{sec:appr-hilb-pair}.
Theoretical results are then developed to evaluate the impact of the Thiran's
\emph{common-factor} degree $L$ on the analytic property of the derived complex
wavelet. In Section~\ref{sec:analytic}, an explicit formula to
quantify the analytic approximation is derived. In addition, we provide a bound
demonstrating the improvement of the analytic property as $L$ increases. These
results apply to all wavelets obtained from FIR filters with Thiran's \emph{common-factor}.
Of particular interest are the orthogonal wavelet bases with perfect
reconstruction.
Section~\ref{sec:perfect_reconstruction} is devoted to proving the
existence of such wavelets arising from filters with Thiran's \emph{common-factor},
which correspond to the wavelets introduced in
\cite{Selesnick2001,Selesnick-thiran}. Finally, in Section~\ref{sec:simul},
some numerical examples illustrate our findings.  All proofs are given in the
Appendices.

\section{Approximate Hilbert pair wavelets}
\label{sec:appr-hilb-pair}
\subsection{Wavelet filters of a Hilbert pair}
Let $\psi_G$ and $\psi_H$ be two real-valued wavelet functions.
Denote by  $\hat\psi_G$ and $\hat\psi_H$ their Fourier transform,
$$
\hat\psi_G(\omega)=\int\psi_G(t)\,\rme^{-\rmi t\omega}\;\rmd\omega\;.
$$
We say that $(\psi_G,\psi_H)$ forms a Hilbert
pair if
$$
\hat\psi_G(\omega)=-\rmi \; \sign(\omega)\hat\psi_H(\omega)\;,
$$
where $\sign(\omega)$ denotes the sign function taking values $-1,0$ and 1 for
$\omega<0$, $\omega=0$ and $\omega>0$, respectively. Then the complex-valued
wavelet $\psi_H (t) + \rmi \psi_G (t)$ is analytic since its Fourier transform is
only supported on the positive frequency semi-axis.

Suppose now that the two above wavelets are obtained from the (real-valued)
low-pass filters $(g_0(n))_{n\in\Z}$ and $(h_0(n))_{n\in\Z}$, using the usual
multi-resolution scheme (see \cite{Daubechies92}).
We denote their z-transforms by $G_0(\cdot)$ and $H_0(\cdot)$, respectively.
In \cite{Selesnick2001} and \cite{Ozkaramanli}, it is established that a
necessary and sufficient condition for $(\psi_G,\psi_H)$ to form a Hilbert pair is to satisfy, for all $\omega\in(-\pi,\pi)$,
\begin{equation}
  \label{eq:hilbert-apir-filters-condition}
  G_0(\rme^{\rmi\omega})=H_0(\rme^{\rmi\omega})\rme^{-\rmi\omega/2}\;.
\end{equation}
Since $\rme^{-\rmi\omega/2}$ takes different values at $\omega=\pi$ and
$\omega=-\pi$, we see that this formula cannot hold if both $G_0$ and $H_0$ are
continuous on the unit circle, which indicates that the construction of Hilbert
pairs cannot be obtained with usual convolution filters and in particular with
finite impulse response (FIR) filters. Hence a strict analytic property for the
wavelet is not achievable for a compactly supported wavelet, which is also a
direct consequence of the Paley-Wiener theorem.

However, for obvious practical reasons, the compact support property of the wavelet and
the corresponding FIR property of the filters must be preserved. Thus the strict analytic condition~\eqref{eq:hilbert-apir-filters-condition}
has to be relaxed into an approximation around the zero frequency,
\begin{equation}
  \label{eq:hilbert-apir-filters-condition-approx}
  G_0(\rme^{\rmi\omega})\sim H_0(\rme^{\rmi\omega})\rme^{-\rmi\omega/2}
  \quad\text{as}\quad\omega\to0\;.
\end{equation}
Several constructions have then been proposed to define approximate Hilbert pair
wavelets, that is, pairs of wavelet functions satisfying the quasi analytic
condition~\eqref{eq:hilbert-apir-filters-condition-approx} \citep{tay2007designing}.
The \emph{common-factor} procedure proposed in \cite{Selesnick-thiran}, is giving one solution to the
construction of approximate Hilbert pair wavelets. This is the focus of the following developments. 

\subsection{The common-factor procedure}

The \emph{common-factor} procedure \citep{Selesnick-thiran} is designed to provide
approximate Hilbert pair wavelets driven by an integer $L\geq1$ and additional
properties relying on a \emph{common factor} transfer function $F$. Namely, the solution reads
\begin{align}
  \label{eqn:CF1}
H_0(z)&=F(z)D_{L}(z)\;,\\
  \label{eqn:CF2}
G_0(z)&=F(z)D_{L}(1/z)z^{-L}\;,
\end{align}
where $D_{L}$ is the $z$ transform of a causal FIR filter of length $L$,
$D_{L}(z)=1+\sum_{\ell=1}^L d(\ell) z^{-\ell}$,
such that
\begin{equation}
\label{eqn:D-approx}
\frac{\rme^{-\rmi\omega L}D_{L}(\rme^{-\rmi\omega})}{D_{L}(\rme^{\rmi\omega})}= \rme^{-\rmi\omega /2}+ O(\omega^{2L+1})  \quad\text{as}\quad\omega\to0\;.
\end{equation}
In \cite{Thiran}, a causal FIR filter satisfying this constraint is defined, the
so-called \emph{maximally flat} solution given by (see also
\cite[Eq~(2)]{Selesnick-thiran}):
\begin{equation}
\label{eqn:selesnick}
d(\ell)=(-1)^n \binom{L}{\ell}\prod_{k=0}^{\ell-1} \frac{1/2-L+k}{3/2+k}, \quad \ell=1,\dots,L.
\end{equation}
The cornerstone of our subsequent results is the following simple expression for
$D_{L}(z)$, which appears to be new, up to our best knowledge. 
\begin{prop}
  \label{prop:D}
  Let $L$ be a positive integer and $D_{L}(z)=1+\sum_{n=1}^L d(n) z^{-n}$
  where the coefficients $(d(n))_n$ are defined by~\eqref{eqn:selesnick}. Then, for all $z\in\C^*$, we
  have
\begin{equation}\label{eq:D-expression-closed-form}
D_{L}(z)=\frac{1}{2(2L+1)}z^{-L}\left[(1+z^{1/2})^{2L+1}+(1-z^{1/2})^{2L+1}\right],
\end{equation}
where $z^{1/2}$ denotes any of the two complex numbers whose squares are equal to $z$.
\end{prop}
Here $\C^*$ denotes the set of all non-zero complex numbers.
\begin{rem} In spite of the ambiguity in the definition of $z^{1/2}$, the
  right-hand side in~\eqref{eq:D-expression-closed-form} is unambiguous
  because, when developing the two factors in the expression between square
  brackets, all the odd powers of $z^{1/2}$ cancel out.
\end{rem}
\begin{rem}
  It is interesting to note that the closed
form expression~\eqref{eq:D-expression-closed-form} of $D_L$ directly implies
the approximation~\eqref{eqn:D-approx}. Indeed, the right-hand side
of~\eqref{eq:D-expression-closed-form} yields
\begin{align*}
  D_{L}(\rme^{\rmi\omega})
  &=\frac1{2(2L+1)}\rme^{-\rmi\omega(L-1/2)/2}
                                                                         \left[(2
                          \cos(\omega/4))^{2L+1}+(-2\rmi
                          \sin(\omega/4))^{2L+1}\right]  \\
  &=\frac{2^{2L}}{2L+1}\rme^{-\rmi\omega(L-1/2)/2}\cos^{2L+1}(\omega/4)+O(\omega^{2L+1})\;.
\end{align*}
It is then straightforward to obtain~\eqref{eqn:D-approx}. 
\end{rem}

\begin{proof}
  See Section~\ref{sec:proofs-analytic}.
\end{proof}

To summarize the \emph{common-factor} approach, we use the following definition.
\begin{defn}[\emph{Common-factor} wavelet filters]
  For any positive integer $L$ and FIR filter with transfer function $F$, a pair of wavelet filters
  $\{H_0,G_0\}$ is called an $L$-approximate Hilbert wavelet filter pair
   with common factor $F$ if it satisfies~\eqref{eqn:CF1}
  and~\eqref{eqn:CF2} with $H_0(1)=G_0(1)=\sqrt{2}$.
\end{defn}
Condition $H_0(1)=G_0(1)=\sqrt{2}$ is equivalent to \begin{equation}
    \label{eq:Fcondition}
    F(1)=\frac{\sqrt2}{D_L(1)}= \sqrt2(2L+1)2^{-2L}\;.
  \end{equation}
A remarkable feature in the choice of the common filter $F$ is that it can be used to ensure additional properties such as an
arbitrary number of vanishing moments, perfect reconstruction or smoothness
properties.

First an arbitrary number $M$ of vanishing moments is set by writing
\begin{equation}
\label{eqn:VM}
F(z)=Q(z)(1+1/z)^M,
\end{equation}
with $Q(z)$ the $z$-transform of a causal FIR filter (hence a real polynomial
of $z^{-1}$).

An additional condition required for the wavelet decomposition is perfect
reconstruction. It is acquired when the filters satisfy the following
conditions (see \cite{PRcondition}):
\begin{align}
\label{eqn:PR-G}
G_0(z)G_0(1/z)+G_0(-z)G_0(-1/z)&=2\;,\tag{PR-G}\\
\label{eqn:PR-H}
H_0(z)H_0(1/z)+H_0(-z)H_0(-1/z)&=2\;\tag{PR-H}.
\end{align}
This condition is classically used for deriving wavelet bases
$\psi_{Gj,k}=2^{j/2}\psi_G(2^j\cdot-k)$ and
$\psi_{Hj,k}=2^{j/2}\psi_H(2^j\cdot-k)$, $j,k\in\Z$, which are 
orthonormal bases of $L^2(\R)$. This will be
investigated in Section~\ref{prop:R}.

\section{Quasi-analyticity of \emph{common-factor} wavelets}
\label{sec:analytic}

We now investigate the quasi-analyticity properties of the complex wavelet
obtained from Hilbert pairs wavelet filters with the \emph{common-factor} procedure. 

Let $(\phi_H(\cdot), \psi_H(\cdot))$ be respectively the father and the mother
wavelets associated with the (low-pass) wavelet filter $H_0$. The transfer
function $H_0$ is normalized so that $H_0(1)=\sqrt2$ (this is implied
by~\eqref{eqn:CF1} and~\eqref{eq:Fcondition}). The  father and  mother
wavelets can be defined through their Fourier transforms as
\begin{align}
\label{eqn:phi:filter}
\hat\phi_H(\omega)&=\prod_{j=1}^\infty \left[2^{-1/2}H_0(\rme^{\rmi 2^{-j}\omega})\right],\\
\label{eqn:psi:filter}
\hat\psi_H(\omega)&=2^{-1/2}H_1(\rme^{\rmi\,\omega/2})\, \hat\phi_H(\omega/2),
\end{align}
where $H_1$ is the corresponding high-pass filter transfer function defined by
$H_1(z)=z^{-1}H_0(-z^{-1})$ (see {\it e.g.} \cite{Selesnick2001}).  We also
denote by $(\phi_G, \psi_G)$ the father and the mother wavelets associated with
the wavelet filter $G_0$. Equations similar to \eqref{eqn:phi:filter} and
\eqref{eqn:psi:filter} hold for $\hat\phi_G$, and $\hat\psi_G$ using $G_0$ and
$G_1$ in place of $H_0$ and $H_1$ (see {\it e.g.}  \cite{Selesnick2001}).

We first give an explicit expression of $\hat\phi_G$ and of
$\hat\psi_G$ with respect to $\hat\phi_H$ and
$\hat\psi_H$.
\begin{thm}\label{th:explicit-phi-psi}
  Let $L$ be a positive integer.  Let $\{H_0,G_0\}$ be an $L$-approximate
  Hilbert wavelet filter pair. Let $(\phi_H, \psi_H)$ denote the father and
  mother wavelets defined by~\eqref{eqn:phi:filter} and~\eqref{eqn:psi:filter}
  and denote $(\phi_G, \psi_G)$ the wavelets defined similarly from the
  filter $G_0$.  Then, we have, for all $\omega\in\R$,
\begin{align}\label{eqn:explicit-phi}
\hat\phi_G(\omega)&=\rme^{\rmi \beta_L(\omega)}\,\hat\phi_H(\omega)\rme^{-\rmi\omega/2}\;,\\
\label{eqn:explicit-psi}
\hat\psi_G(\omega)&=
\rmi\;\rme^{\rmi \eta_L(\omega)}\,\hat\psi_H(\omega)\;.
\end{align}
where
\begin{align}\label{eq:alpha-def}
  \alpha_L(\omega)&= 2(-1)^L\,\atan\left(\tan^{2L+1}(\omega/4)\right)\;,\\
  \label{eq:beta-def}
  \beta_L(\omega)&=\sum_{j=1}^{\infty} \alpha_L(2^{-j}\omega)\;,\\
  \label{eq:eta-def}
  \eta_L(\omega)&=-\alpha_L(\omega/2+\pi)+\beta_L(\omega/2)\;.
\end{align}
In~\eqref{eq:alpha-def}, we use the convention $\atan(\pm\infty)=\pm\pi/2$ so
that $\alpha_L$ is well defined on $\R$.
\end{thm}
\begin{proof}
  See Section~\ref{sec:proofs-analytic}.
\end{proof}
Following Theorem~\ref{th:explicit-phi-psi}, we can write, for all $\omega\in\R$,
\begin{equation}\label{eq:uL}
\hat\psi_H(\omega)+\rmi\,\hat\psi_G(\omega)=\left(1-\rme^{\rmi \eta_L(\omega)}\right)\hat\psi_H(\omega)\;.
\end{equation}
This formula shows that the quasi-analytic property and the Fourier localization
of the complex wavelet $\psi_H+\rmi\,\psi_G$ can be respectively
described by
\begin{enumerate}[label=(\alph*)]
\item\label{item:3} how close the function $1-\rme^{\rmi \eta_L}$ is to the
  step function $2\1_{\R_+}$ (or $-\rme^{\rmi \eta_L}$ to the
  sign function);
\item\label{item:4} how localized the (real) wavelet $\psi_H$ is in the Fourier domain. 
\end{enumerate}
Property~\ref{item:4} is a well known feature of wavelets usually described by
the behavior of the wavelet at frequency 0 (e.g. $M$ vanishing moments implies a
behavior in $O(|\omega|^M)$) and by the polynomial decay at high frequencies.
This behavior depends on the wavelet filter (see~\cite{villemoes1992energy,eirola1992sobolev,ojanen2001orthonormal}) and a numerical study of property~\ref{item:4} is provided in Section~\ref{sec:simul}.

Note that, remarkably,  property~\ref{item:3}, \textbf{only depends on} $L$. Figure~\ref{fig:approx} displays the
function $1-\rme^{\rmi \eta_L}$ for various values of $L$. It illustrates the
fact that as $L$ grows, $1-\rme^{\rmi \eta_L}$ indeed gets closer and closer to
the step function $2\1_{\R_+}$.  We can actually prove the following result
which bounds how close the Fourier transform of the wavelet $\psi_H+\rmi\,\psi_G$ is to $2\1_{\R_+}\,\hat\psi_H$.

Denote, for all $\omega\in\R$ and $A\subset\R$, the distance of $\omega$ to $A$
by
\begin{equation}
\label{eqn:delta}
\delta(\omega,A)=\inf\left\{|\omega-x|~:~x\in A\right\}\;.
\end{equation}

\begin{thm}\label{th:main-result-analyticity}
  Under the same assumptions as Theorem~\ref{th:explicit-phi-psi}, we have, for
  all $\omega\in\R$,
  $$
  \left|\hat\psi_H(\omega)+\rmi\,\hat\psi_G(\omega) -
    2\1_{\R_+}(\omega)\,\hat\psi_H(\omega)\right|= U_L(\omega)
  \left|\hat\psi_H(\omega)\right|\;,
  $$
  where $U_L$ is a $\R\to[0,2]$ function satisfying, for
  all $\omega\in\R$,
    \begin{equation}
    \label{eq:final-analycity-result}
U_L(\omega) \leq
    2\sqrt{2}
\left(\log_2\left(\frac{\max(4\pi,|\omega|)}{2\pi}\right)+2\right)\,\left(1-\frac{\delta(\omega,4\pi\Z)}{\max(4\pi,|\omega|)}\right)^{2L+1}
 \;.
\end{equation}
\end{thm}
\begin{proof}
  See Section~\ref{sec:proofs-analytic}.
\end{proof}
This result provides a control over the difference between the Fourier
transform $\hat\psi_H+\rmi\,\hat\psi_G$ of the complex wavelet and
the Fourier transform $2\1_{\R_+}\,\hat\psi_H$ of the analytic signal
associated to $\psi_H$. In particular, as
$L\to\infty$, the relative difference
$U_L={\left|\hat\psi_H+\rmi\,\hat\psi_G-2\1_{\R_+}\hat\psi_H\right|}/{\left|\hat\psi_H\right|}$
converges to zero exponentially fast on any compact
subsets that do not intersect $4\pi\Z$. 

\begin{figure}[!ht]
\begin{center}
\includegraphics[scale=0.4]{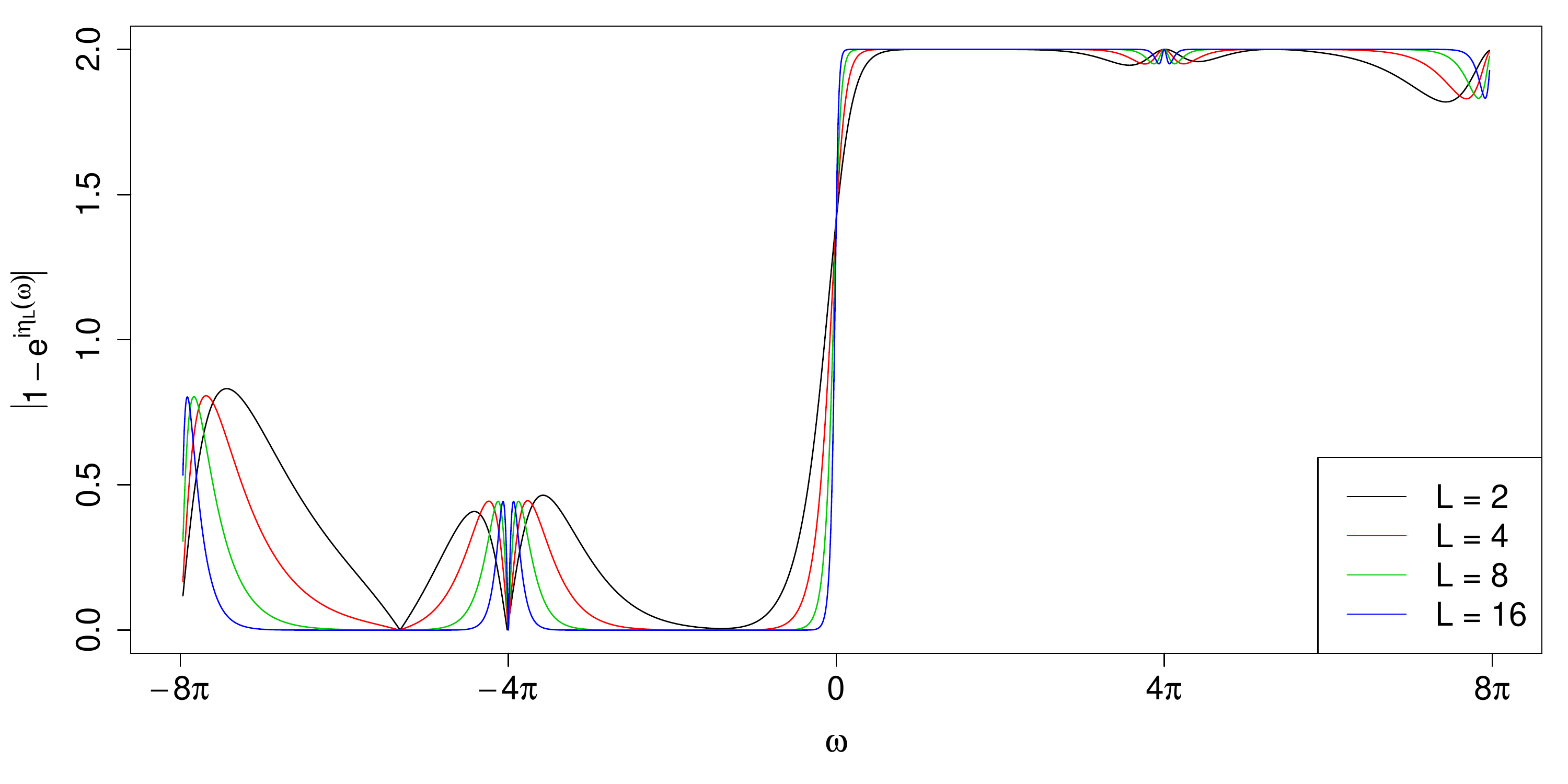}
\end{center}
\caption{Plots of the function
  $\omega\mapsto\lvert  1-\rme^{\rmi \eta_L(\omega)}\rvert$ for $L=2$, 4, 8,  16.}
\label{fig:approx}
\end{figure}
 
\section{Solutions with perfect reconstruction}
\label{sec:perfect_reconstruction}

Let us now follow the path paved by \cite{Selesnick-thiran} to select $Q$
appearing in the factorization~\eqref{eqn:VM} of the common factor $F$ to
impose $M$ vanishing moments.  First observe that,
under~\eqref{eqn:CF1},~\eqref{eqn:CF2} and~\eqref{eqn:VM},
the perfect reconstruction conditions~\eqref{eqn:PR-G} and~\eqref{eqn:PR-H}
both follow from 
\begin{align}
  \label{eqn:PR}
  R(z)S(z)+R(-z)S(-z)=2 \;,
\end{align}
where we have set $R(z)=Q(z)Q(1/z)$ and  $S(z)=(2+z+1/z)^M D_{L}(z)D_{L}(1/z)$. 

To achieve~\eqref{eqn:PR}, the following procedure is proposed in
\cite{Selesnick-thiran}, which follows the approach in \cite{Daubechies92}
adapted to the common factor constraint in~\eqref{eqn:CF1}.
\begin{enumerate}[label=\textbf{Step \arabic*}]
\item\label{item:Rstep} Find $R$ with finite, real and symmetric impulse
  response satisfying~\eqref{eqn:PR}.
\item\label{item:Qstep}  Find a real polynomial $Q(1/z)$ satisfying the factorization $R(z)=Q(z)Q(1/z)$.
\end{enumerate}
However, in \cite{Selesnick-thiran}, the existence of solutions $R$ and
$Q$ is not proven, although numerical procedures indicate that solutions
can be exhibited.  We shall now fill this gap and show the existence of
such solutions for any integers $M,L\geq1$.

We first establish the set of solutions for $R$.
\begin{prop}
  \label{prop:R}
  Let $L$ and $M$ be two positive integers.
  Let $D_{L}$ be defined as in Proposition~\ref{prop:D} and let
  $S(z)=(2+z+1/z)^M D_{L}(z)D_{L}(1/z)$. Then the two following assertions
  hold.
  \begin{enumerate}[label=(\roman*)]
  \item\label{item:1}   There exists a unique real polynomial $r$
    of degree
  at most $M+L-1$ such
  that $R(z)=r\left(\frac{2+z+1/z}{4}\right)$ 
satisfies~\eqref{eqn:PR}
  for all $z\in\C^*$.
\item\label{item:2} For any real polynomial $p$, the function
  $R(z)=p\left(\frac{2+z+1/z}{4}\right)$ satisfies~\eqref{eqn:PR} on
  $z\in\C^*$ if and only if it satisfies
\begin{equation}
  \label{eq:P-def-with-q}
p(y)=r(y)+s(1-y)\,q(y)\;,
\end{equation}
where
\begin{equation}
  \label{eq:PS-def}
  s(y)=y^M\sum_{n=0}^L{{2L+1}\choose{2n}}y^n\;,
\end{equation}
and $q$ is any real polynomial satisfying $q(1-y)=-q(y)$.
  \end{enumerate}
\end{prop} 
\begin{proof}
  See Section~\ref{sec:proofs-perfect-reconstruction}.
\end{proof}

Proposition~\ref{prop:R} provides a justification of~\ref{item:Rstep}. 
In particular, a natural candidate for~\ref{item:Rstep} is 
  $R(z)=r\left(\frac{2+z+1/z}{4}\right)$. Now, by the Riesz Lemma (see
e.g. \cite[Lemma 6.1.3]{Daubechies92}), the factorization of~\ref{item:Qstep} holds if and
only if $R(z)$ takes its values in $\R_+$ on the unit circle $\{z\in\C~:~\lvert z\rvert=1\}$, or
equivalently $r(y)\geq0$ for all $y\in[0,1]$. Although easily verifiable in
practice (using a numerical computation of the roots of $r$), checking this
property theoretically for all integers $L,M$ is not yet achieved.

Nevertheless we next prove that~\ref{item:Qstep} can always be carried out for
any $L,M\geq1$, at least by modifying $r$ into a polynomial $p$ of the
form~\eqref{eq:P-def-with-q} with a conveniently chosen $q$.  

\begin{thm}\label{thm:Q}
  Let $L$ and $M$ be two positive integers and let $r$ and $s$ be the
  polynomials defined as in Proposition~\ref{prop:R}. Then there exists a real
  polynomial $q$ such that $R(z)=[r+s\,q]\left(\frac{2+z+1/z}{4}\right)$ is a
  solution of~\eqref{eqn:PR} and satisfies the factorization $R(z)=Q(z)Q(1/z)$
  where $Q(1/z)$, real polynomial of $z$, does not vanish on the unit
  circle.
\end{thm}
\begin{proof}
  See Section~\ref{sec:proofs-perfect-reconstruction}.
\end{proof}
 
Proposition~\ref{prop:R} and Theorem~\ref{thm:Q} allows one to carry out the
usual program to the construction of compactly supported orthonormal wavelet
bases, as described in \cite{Daubechies92}. Hence we get the following.
\begin{cor}
    Let $L$ and $M$ be two positive integers. Let $Q$ be as in
    Theorem~\ref{thm:Q}. Define $F$ as in~\eqref{eqn:VM} and let $\{H_0,G_0\}$
    be the $L$-approximate 
  Hilbert wavelet filter pair associated to $F$. Then the wavelet bases
  $(\psi_{H,j,k})$ and   $(\psi_{G,j,k})$ are orthonormal bases of $L^2(\R)$. 
\end{cor}

Observe that Theorem~\ref{thm:Q} states the existence of the polynomial $Q$ but
does not define it in a unique way. We explain why in the following remark.
\begin{rem}\label{rem:solut-minimal-degree}
  Since $r$ in Proposition~\ref{prop:R} is defined uniquely, it follows that,
  if we require that all the roots of $Q$ are inside the unit circle, there is
  \emph{at most one solution} for $Q$ with degree at most $K=M+L-1$, which
  correspond to the case $q=0$. This solution, when it exists, is usually
  called the \emph{minimal phase, minimal degree} solution.  However we were
  not able to prove that $r$ does not vanish on $[0,1]$, which is a necessary
  and sufficient condition to obtain such a minimal degree solution for
  $Q$. Hence we instead prove the existence of solutions for $Q$
  by allowing $q$ to be non-zero.
\end{rem}

\section{Numerical computation of approximate Hilbert wavelet filters}
\label{sec:simul}

\subsection{State of the art}

Let $M$ and $L$ be positive integers.  Then, by Theorem~\ref{thm:Q}, we can
define the polynomial $Q$ and derive from its coefficients the impulse response
of the corresponding $L$-approximate Hilbert wavelet filter pair with $M$
vanishing moments and perfect reconstruction. 

We now discuss the numerical computation of the coefficients of $Q$ in the case
where the polynomial $r$ defined by Proposition~\ref{prop:R} does not vanish on
$[0,1]$. Indeed suppose that one can obtain a numerical computation of this
polynomial $r$. Then the roots of $r$ can also be computed by a numerical
solver and, as explained in Remark~\ref{rem:solut-minimal-degree}, if they do
not belong on $[0,1]$ (which has to be checked taking into account the possible
numerical errors), it only remains to factorize $R(z)=r((2+z+1/z)/4)$ into $Q(z)Q(1/z)$ by
 separating the roots conveniently. Taking all roots of modulus inferior to 1
leads to ``mid-phase'' wavelets.  There are other ways of factorizing $R$,
namely ``min-phase'' wavelets, see \cite{Selesnick-thiran}, leading to wavelets
with Fourier transform of the same magnitude but with
different phases. This difference can be useful in some multidimensional
applications where the phase is essential.

Hence the computation of the wavelet filters boils down to the numerical
computation of the polynomial $r$ defined by Proposition~\ref{prop:R}. In
 \cite{Selesnick-thiran}, this computation is achieved by using the
following algorithm.
\begin{itemize}
\item Let $s_1=(\binom{k}{2M})_{k=0,\dots,2M}$ and $s_2=\begin{pmatrix}
d_L(0)&\dots&d_L(L)
\end{pmatrix} \star \begin{pmatrix}
d_L(L)&\dots&d_L(0)
\end{pmatrix}$, where $\star$ denotes the convolution for sequences. Then
$S(z)=(2+z+1/z)^M D_L(z)D_L(1/z)=\sum_{n=0}^{2(M+L)}s(n)z^{n-(M+L)}$ with
$s=s_1\star s_2$. The filter $s$ has length $2(M+L)+1$.
\item The filter $r$ is such that $s\star r$ is half-band. Let $T$ denote the
  Toeplitz matrix associated with $\begin{pmatrix} 0 \dots 0 & s\end{pmatrix}$,
  vector of length $4(M+L)+1$, that is, $T_{k,j}=s(1+(k-j))$ if $0\leq k-j\leq
  2(M+L)$ and $T_{k,j}=0$ else. We introduce $C$ the matrix obtained by keeping
  only the even rows of $T$, which has size $(2(M+L)-1)\times (2(M+L)-1)$. Then
  $r$ is the solution of the equation
  \begin{equation}
    \label{eq:selesnick-algo-inversion-step}
  Cr=b  
  \end{equation}
  with $b=\begin{pmatrix}
0&\dots&0&1&0&\dots&0
\end{pmatrix}$ a $2(M+L)-1$-vector with a $1$ at the middle ({\it i.e.} at $(M+L)$-th position). 
\end{itemize}
We implemented this linear inversion method but it turned out that the corresponding linear
equation is ill posed for too high values of $M$ and $L$ (for instance
$M=L=7$). For smaller values of $L$ and $M$, we recover the wavelet filters of the
\texttt{hilbert.filter} program of the \texttt{R}-package \texttt{waveslim}
computed only for $(M,L)$ equal to  (3,3), (3,5), (4,2) and (4,4), see
\cite{whitcher15-waveslim-r-package}.

\subsection{A recursive approach to the computation of the Bezout minimal
  degree solution}

We propose now a new method for computing the $L$-approximate \emph{common-factor}
wavelet pairs with $M$ vanishing moments under the  perfect reconstruction
constraint. As explained previously, this computation reduces to determining
the coefficients of the polynomial $r$ defined in Proposition~\ref{prop:R}.
Our approach is intended as an alternative to the linear system resolution
step of the approach proposed in \cite{Selesnick-thiran}.
Since our algorithm is recursive, to avoid any ambiguity, we add the subscripts
$L,M$ for denoting the polynomials $r$ and $s$ appearing in~\ref{prop:R}. That
is, we set  $$s_{L,M}(y)=y^M\sum_{n=0}^L{{2L+1}\choose{2n}}y^n$$ and $r_{L,M}$ is
the unique polynomial of degree at most $M+L-1$ satisfying the Bezout equation
\begin{enumerate}[label=${[B(L,M)]}$,labelsep=1.4cm,align=left]
\item \centering$r_{L,M}(1-y) s_{L,M}(1-y)+r_{L,M}(y) s_{L,M}(y)=(2L+1)^2\,2^{-2L-2M+1}\;.$
\end{enumerate}
We propose to compute $r_{L,M}$ for all $L\geq1$, $M\geq0$ by using the following result.
\begin{prop}\label{pro:expr-rL0}
Let $L\geq1$.  Define
\begin{equation}
  \label{eq:s-roots}
y_{k,L}=-\tan^2\left(\frac{\pi(2k+1)}{2(2L+1)}\right)\,,\qquad k\in \{0,\cdots,L-1\}\;.  
\end{equation}
Then the solution $r_{L,0}$ of the Bezout equation $[B(L,0)]$ is
given by
\begin{equation}\label{eq:interp-rLM}
  r_{L,0}(y)=(2L+1)^2\,2^{-2L+1}\sum_{k=0}^{L-1}
\frac{\prod_{m\neq k} (y-(1-y_{m,L}))}{s_{L,0}(1-y_{k,L})\prod_{m\neq k} (y_{m,L}-y_{k,L})}\;.
\end{equation}
Moreover, for all $M\geq1$, we have the following relation between the solution
of $[B(L,M)]$ and that of $[B(L,M-1)]$:
\begin{equation}
  \label{eq:recursiv-formula-r}
4 \,y\, r_{L,M}(y)=r_{L,M-1}(y)\,-\,2^{-2L}\,r_{L,M-1}(0)\,(1-2y)\,s_{L,M-1}(1-y)\;.
\end{equation}
\end{prop}
\begin{proof}
  See Appendix~\ref{sec:defin-polyn-r}.
\end{proof}
This result provides a recursive way to compute $r_{L,M}$ by starting with
$r_{L,0}$ using the interpolation formula~(\ref{eq:interp-rLM}) and then using
the recursive formula~(\ref{eq:recursiv-formula-r}) to compute
$r_{L,1},r_{L,2},\dots$ up to $r_{L,M}$. In contrast to the method of
\cite{Selesnick-thiran} which consists in solving a (possibly ill posed) linear
system, this method is only based on product and composition of polynomials. 

\subsection{Some numerical result on smoothness and  analyticity}
We now provide some numerical results on the quality of the analyticity of the
$L$-approximated Hilbert wavelet. All the numerical computations have been
carried out by the method of
\cite{Selesnick-thiran} which seems to be the one used by practitioners (as in
the software of \cite{whitcher15-waveslim-r-package}).
Recall that as established in
Theorem~\ref{th:main-result-analyticity}, for all $\omega\in\R$,
  $$
  \left|\hat\psi_H(\omega)+\rmi\,\hat\psi_G(\omega) -
    2\1_{\R_+}(\omega)\,\hat\psi_H(\omega)\right|= U_L(\omega)
  \left|\hat\psi_H(\omega)\right|\;,
  $$
  where $U_L$ is displayed in Figure~\ref{fig:regularity}. Thus the quality of
  analyticity relies on the behavior of $U_L$ but also of
  $\hat\psi_H(\omega)$. First, $\hat\psi_H(\omega)$ goes to $0$ when
  $\omega\to 0$ thanks to the property of $M$ vanishing moments given
  by~\eqref{eqn:VM}. Secondly, 
  $\lvert\hat\psi_H(\omega)\rvert$ decays to zero as $|\omega|$ goes to
  infinity. This last point is verified numerically, by the estimation of
  the Sobolev exponents of $\psi_H$ using \cite{ojanen2001orthonormal}'s
  algorithm. Values are given in Table~\ref{tab:sobolev}. For $M>1$ Sobolev
  exponents are greater than 1. Notice that ``min-phase'' and ``mid-phase''
  factorizations of $R$ have the same exponents since the methods do not change
  the magnitude of $\hat \psi_H+\rmi\hat\psi_G$.

\begin{table}[!h]
  \caption{Sobolev exponent estimated for $\psi_H$ functions. Dots correspond
    to configurations where numerical instability occurs in the numerical inversion
    of~(\ref{eq:selesnick-algo-inversion-step}).}
\label{tab:sobolev}
\begin{center}
\begin{tabular}{l|cccccccc}
M $\backslash$ L & 1 & 2 &3 & 4& 5 & 6 &7 &8 \\
\hline
1  & 0.60  & 0.72  & 0.81  & 0.89  & 0.94  & 0.98  & 0.99  & 1.00  \\  
2  & 1.11  & 1.23  & 1.34  & 1.44  & 1.54  & 1.63  & 1.73  & 1.82  \\  
3  & 1.52  & 1.64  & 1.74  & 1.83  & 1.92  & 2.01  & 2.09  & 2.17  \\  
4  & 1.87  & 1.98  & 2.07  & 2.16  & 2.24  & 2.32  & 2.40  & 2.48  \\  
5  & 2.19  & 2.29  & 2.37  & 2.45  & 2.53  & 2.60  & 2.68  & $\cdot $  \\  
6  & 2.48  & 2.57  & 2.65  & 2.72  & 2.80  & 2.87  & $\cdot $  & $\cdot$  \\  
7  & 2.74  & 2.83  & 2.91  & 2.98  & 3.05  & 3.12  & $\cdot $  & $\cdot$  \\  
8  & 3.00  & 3.09  & 3.16  & 3.23  & 3.29  & $\cdot $  & $\cdot $  & $\cdot$  \\ 
\end{tabular}
\end{center}
\end{table}

Figure~\ref{fig:regularity} displays the overall shapes of the Fourier
transforms $\hat\psi_H$ of orthonormal wavelets with \emph{common-factor} for various
values of $M$ and $L$. Their quasi-analytic counterparts
$\hat\psi_H(\omega)+\rmi\,\hat\psi_G(\omega)$ are plotted below in the same
scales. It illustrates the satisfactory quality of analityc approximation.

\begin{figure}[h]

\includegraphics[scale=0.43]{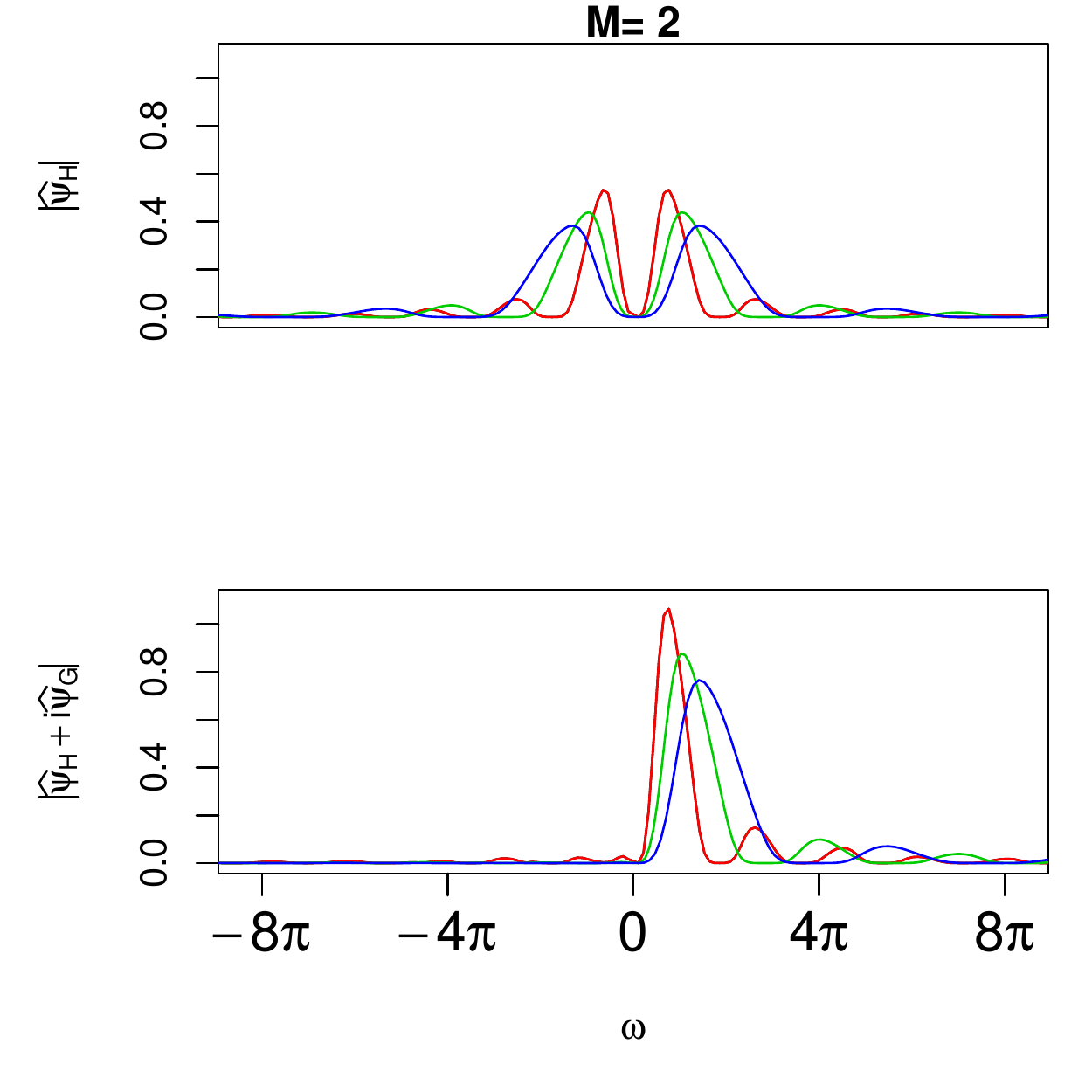}\includegraphics[scale=0.43]{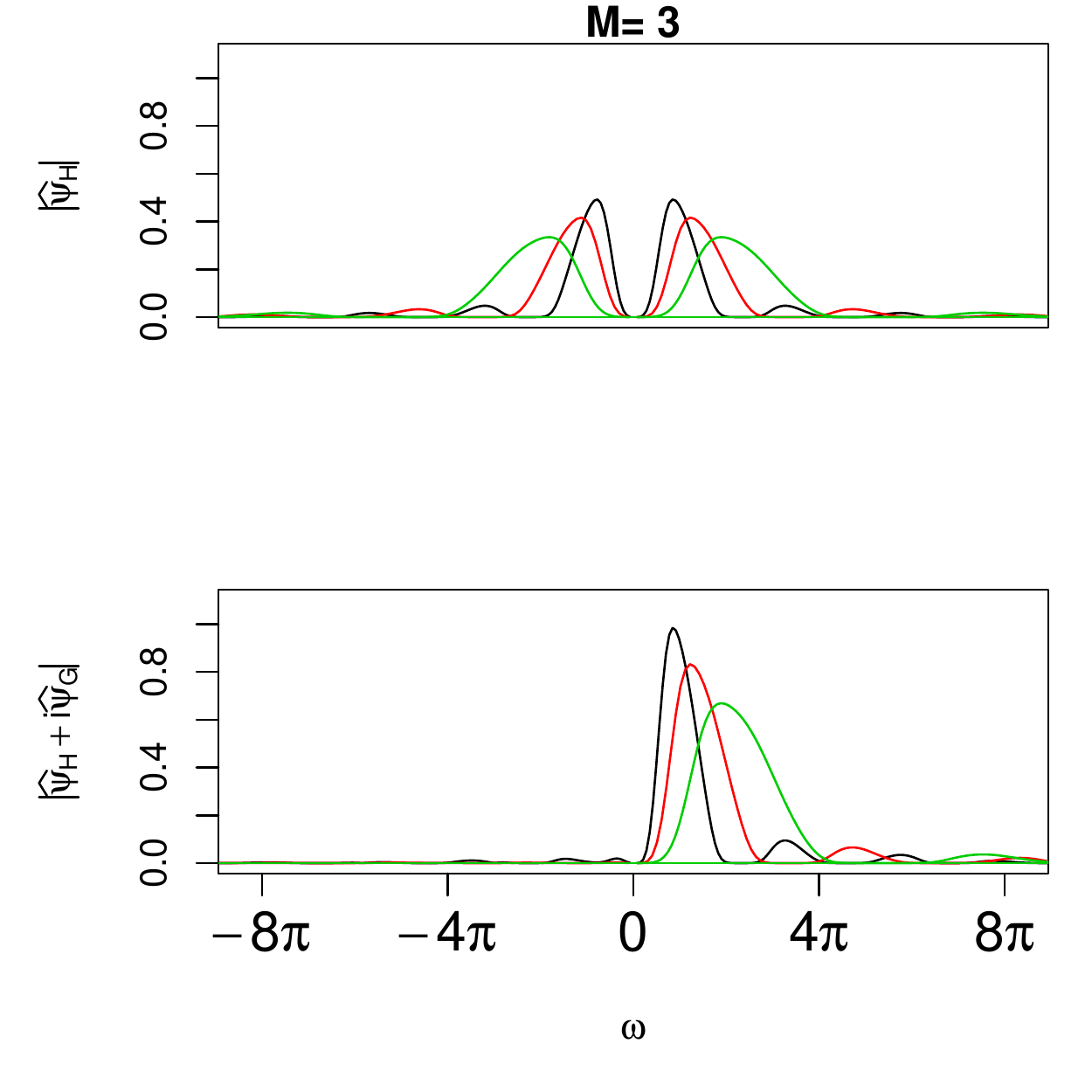}\includegraphics[scale=0.43]{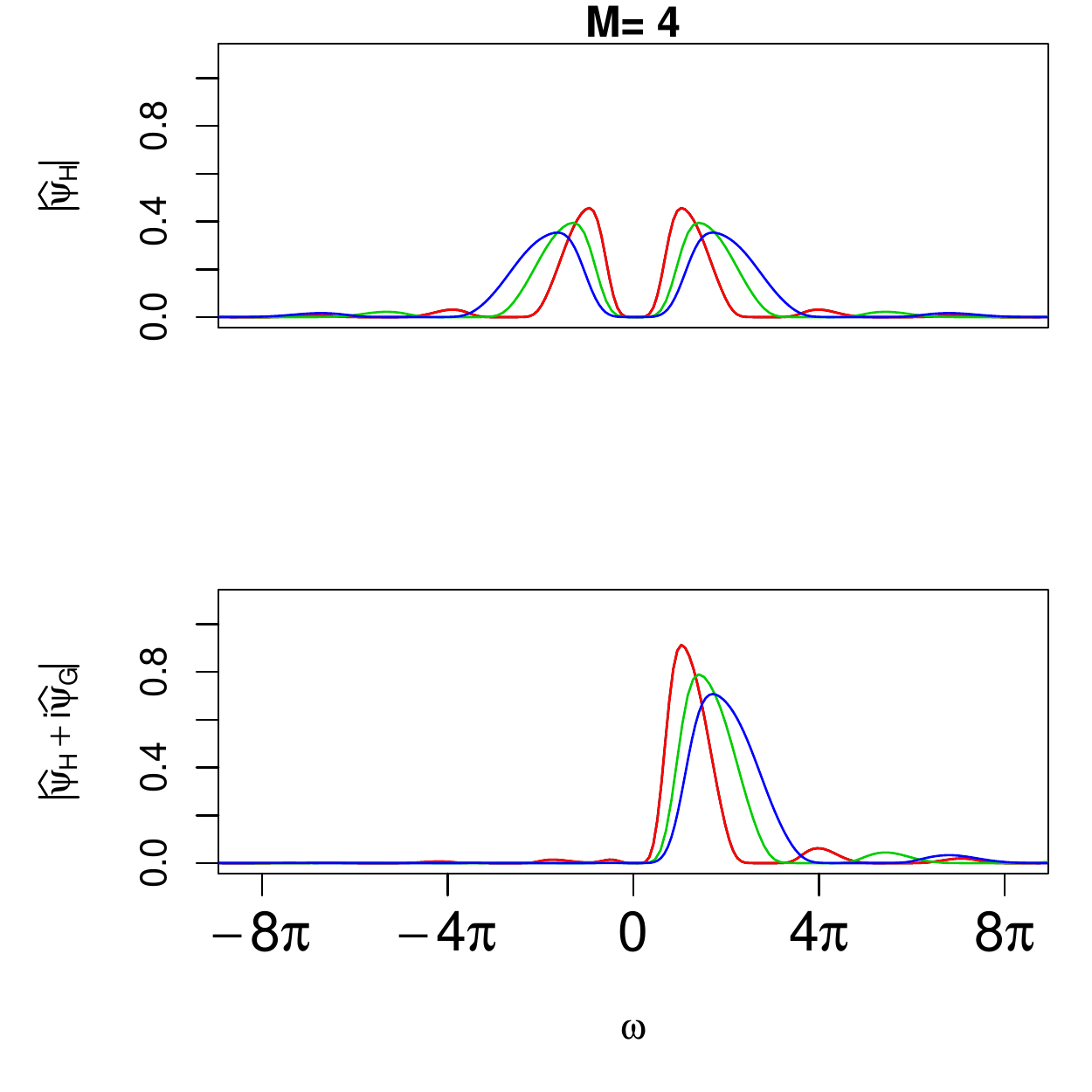}

\caption{Top row: Plots of $\lvert\hat\psi_H\rvert$ for $M=2$(left), 3 (center), 4 (right)
  and $L=2$ (black), 4 (red), 8 (green). Bottom row: same for  $\lvert\hat\psi_H+\rmi\hat\psi_G\rvert$.}
\label{fig:regularity}
\end{figure}

\cite{tay2006orthonormal} propose two objective measures of quality
based on the spectrum, 
\[
E_1 = \frac{\max \{\lvert\hat\psi_H(\omega)+\rmi\hat\psi_G(\omega)\rvert ,\,\omega<0 \}}{\max \{\lvert\hat\psi_H(\omega)+\rmi\hat\psi_G(\omega)\rvert ,\,\omega>0 \}}\,\text{~and~}
E_2 = \frac{\int_{\omega <0} \lvert\hat\psi_H(\omega)+\rmi\hat\psi_G(\omega)\rvert^2\,d\omega}{\int_{\omega >0} \lvert\hat\psi_H(\omega)+\rmi\hat\psi_G(\omega)\rvert^2\,d\omega}\,.
\]
Numerical values of $E_1$ and $E_2$ are computed using numerical evaluations of
$\hat\psi_H$ on a grid, and, concerning $E_2$, using Riemann sum approximations
of the integrals.  Such numerical computations of $E_1$ and $E_2$ are displayed
in Figure~\ref{fig:error} for various values of $M$ and $L$. The functions $E_1$ and $E_2$ are 
decreasing with respect to $L$ (which corresponds to the behaviour of
$U_L$). They are also decreasing with respect to $M$ (through the faster decay
of $\hat\psi_H$ around zero and infinity). Moreover, the values illustrate the good analyticity quality
of \emph{common-factor} wavelets. For example, values appear to be lower than those of
approximate analytic wavelets based on Bernstein polynomials given in
\cite{tay2006orthonormal}.
\begin{figure}[h]

\begin{center}
\includegraphics[scale=0.6]{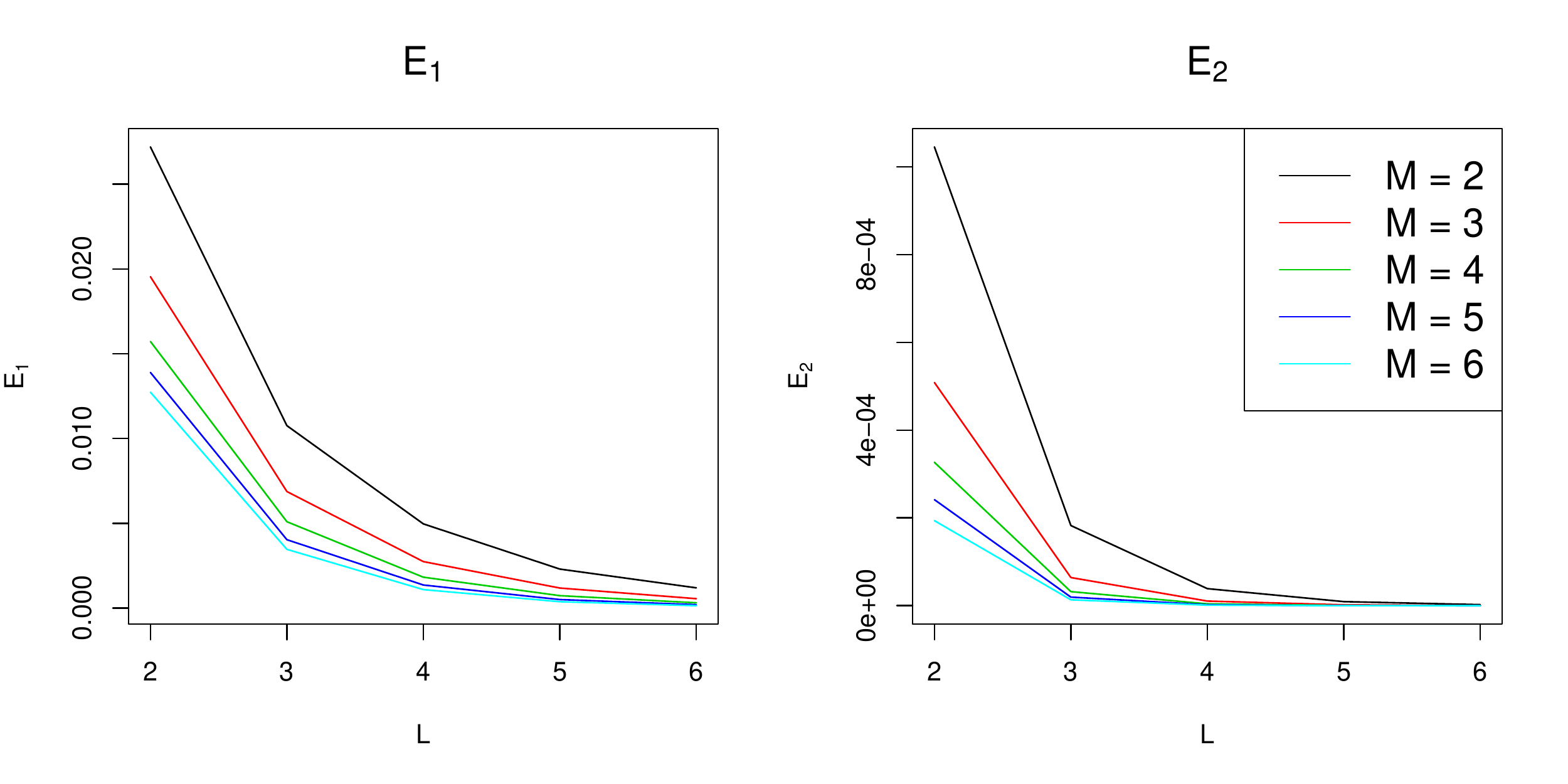}
\end{center}
\caption{Plot of $E_1$ and $E_2$ with respect to $L$ for different values of $M$.}
\label{fig:error}

\end{figure}

\section{Conclusion}

Approximate Hilbert pairs of wavelets are built using the \emph{common-factor}
approach. Specific filters are obtained under perfect reconstruction
conditions.  They depend on two integer parameters $L$ and $M$ which correspond
respectively to the order of the analytic approximation and the number of null
moments. We demonstrate that the construction of such wavelets is valid by
proving their existence for any parameters $L,M \geq 1$.  Our main contribution
in this paper is to provide an exact formula of the relation between the
Fourier transforms of the two real wavelets associated to the filters. This
expression allows us to evaluate the analyticity approximation of the wavelets,
{\it i.e.} to control the presence of energy at the negative frequency. This
result may be useful for applications, where the approximated analytic
properties of the wavelet have to be optimized, in addition to the usual
localization in time and frequency. Numerical simulations show that these
wavelets are easy to compute for not too large values of $L$ and $M$, and
confirm our theoretical findings, namely, that the analytic approximation
quickly sharpens as $L$ increases.

\section{Acknowledgements}

This research did not receive any specific grant from funding agencies in the public, commercial, or not-for-profit sectors.

\appendix

\section{Proofs of Section~\ref{sec:analytic}}
\label{sec:proofs-analytic}

\subsection{Proof of Proposition~\ref{prop:D}}

\begin{proof}[Proof of~\eqref{eq:D-expression-closed-form}]
Notice that $d(L)^{-1}z^LD_{L}(z)=\sum_{n=0}^L \frac{d(L-n)}{d(L)} z^n$ and that for all $n=0,\ldots,L-1$,
\begin{align*}
\frac{d(L-n)}{d(L)}&=\binom{L}{n}\prod_{\ell=L-n}^{L-1}\frac{2\ell+3}{2L-2\ell-1} \\
&=\binom{L}{n}\left(\prod_{k=1}^{n}{(2k-1)}\right)^{-1}\prod_{\ell=L-n+1}^{L}{(2\ell+1)} \\
&=\frac{L!}{n!(L-n)!}\frac{2^n n!}{(2n)!}\frac{(2L+1)!}{2^L L!}\frac{2^{L-n}(L-n)!}{(2L-2n+1)!} \\
&= \binom{2L+1}{2n}
\end{align*}
It is then easy to check that
\[
d(L)^{-1}z^L D_{L}(z)=\frac{1}{2}\left((1+z^{1/2})^{2L+1}+(1-z^{1/2})^{2L+1}\right).
\]
The fact that $d(L)=1/(2L+1)$ concludes the proof.
\end{proof}

\subsection{Technical results on $D$}

We first establish the following result, which will be useful to handle ratios
with $D_{L}(\rme^{i\omega})$.
\begin{lem}
\label{lem:minD}
  Let $L$ be a positive integer. Define $D_{L}$ as in
  Proposition~\ref{prop:D}. Then $D_{L}(z)$ does not vanish on the unit circle
  ($|z|=1$) and
$$
\min_{z\in\C~:~|z|=1} \left|D_{L}(z)\right|=|D_{L}(-1)|=\frac{2^L}{2L+1}
<\max_{z\in\C~:~|z|=1} \left|D_{L}(z)\right|=|D_{L}(1)|=\frac{2^{2L}}{2L+1}\;.
$$
\end{lem}
\begin{proof}
  Since $D_{L}(1/z)$ is a real polynomial of $z$, we have for all $z\in\C$ such
  that $|z|=1$, $\left|D_{L}(z)\right|^2=D_{L}(z)D_{L}(1/z)$. Moreover, as shown in the
  proof of Proposition~\ref{prop:R}, if $z=\rme^{2\rmi\theta}$ with
  $\theta\in\R$, then $\left|D_{L}(z)\right|^2$ reads as in~\eqref{eq:DDeq}, which
  is minimal and maximal for $\cos(\theta)=0$ and $1$, respectively.
\end{proof}
We now study $z^{-L}\frac{D_{L}(1/z)}{D_{L}(z)}$ on the circle.
\begin{lem}
\label{lem:ratioD}
For all $z=\rme^{\rmi\omega}$ with $\omega\in\R$, we have
\begin{align}
\rme^{-\rmi\omega
  L}\frac{D_{L}(\rme^{-\rmi\omega})}{D_{L}(\rme^{\rmi\omega})}&=\rme^{-\rmi\omega/2+\rmi \alpha_L(\omega)}\;,
\end{align}
where $\alpha_L$ is the function defined on $\R$ by~\eqref{eq:alpha-def}. 
\end{lem}
\begin{proof}
Observe that, for all $z\in\C^*$, denoting by $z^{1/2}$ any of the two roots of $z$,
\begin{align*}
z^{-L}\frac{D_{L}(1/z)}{D_{L}(z)}&=z^{L}\frac{(1+z^{-1/2})^{2L+1}+(1-z^{-1/2})^{2L+1}}{(1+z^{1/2})^{2L+1}+(1-z^{1/2})^{2L+1}}\\
&=z^{-1/2}\frac{(1+z^{1/2})^{2L+1}+(z^{1/2}-1)^{2L+1}}{(1+z^{1/2})^{2L+1}-(z^{1/2}-1)^{2L+1}}
\end{align*}
Set now $z=\rme^{\rmi\omega}$. We deduce that
\[
\rme^{-\rmi\omega L}\frac{D_{L}(\rme^{-\rmi\omega})}{D_{L}(\rme^{\rmi\omega})}=
\rme^{-\rmi\omega/2}\frac{\rme^{\rmi\omega(2L+1)/4}\cos(\omega/4)^{2L+1}(1+\rmi(-1)^L\tan(\omega/4)^{2L+1})}{\rme^{\rmi\omega(2L+1)/4}\cos(\omega/4)^{2L+1}(1-\rmi(-1)^L\tan(\omega/4)^{2L+1})}.
\]
The result then follows from the classical result
$\frac{1+\rmi a}{1-\rmi a}= \rme^{2\,\rmi\,\atan(a)}$ with here $a=(-1)^L\tan(\omega/4)^{2L+1}$.
\end{proof}

\subsection{Proof of Theorem \ref{th:explicit-phi-psi}}

\begin{proof}[\textbf{Proof of equality \eqref{eqn:explicit-phi}}]
Equation~\eqref{eqn:phi:filter} provides the relation between $\hat\phi_H$ and $H_0$. 
The same relation holds between $\hat\phi_G$ and $G_0$. It follows with
 Lemma~\ref{lem:minD},~\eqref{eqn:CF1}
and~\eqref{eqn:CF2}, that, for all $\omega\in\R$,
\[
\hat\phi_G(\omega)=\hat\phi_H(\omega)\prod_{j=1}^\infty \left[\rme^{-\rmi\omega2^{-j}L}\frac{D_{L}(\rme^{-\rmi\omega2^{-j}})}{D_{L}(\rme^{\rmi\omega2^{-j}})}\right]\;.
\]
Applying Lemma~\ref{lem:ratioD}, we get that, for all $\omega\in\R$,
\begin{align*}
\hat\phi_G(\omega)&=\hat\phi_H(\omega)\prod_{j=1}^\infty \rme^{-\rmi\omega 2^{-j}/2+\rmi \alpha_L(\omega 2^{-j})}\\
&=\hat\phi_H(\omega)\exp\left(-\rmi\omega/2 \sum_{j=1}^\infty 2^{-j}+\rmi \sum_{j=1}^\infty \alpha_L(\omega 2^{-j})\right)\;.
\end{align*}
We thus obtain~\eqref{eqn:explicit-phi} using the definition of $\beta_L$ given by~\eqref{eq:beta-def}.
\end{proof}

\begin{proof}[\textbf{Proof of equality \eqref{eqn:explicit-psi}}]
First observe that the 
relation between the high-pass filters $G_1$ and $H_1$ follows from that
between the low-pass filter $G_0$ and $H_0$, namely
\[
G_1(z)=(-z)^L\frac{D_{L}(-z)}{D_{L}(-1/z)}H_1(z).
\]
The relationship between $\hat\psi_G$ and $\hat\phi_G$ is given
by~\eqref{eqn:psi:filter} (exchanging $G$ and $H$), yielding,
for all $\omega\in\R$,
\[
\hat\psi_G(\omega)=2^{-1/2}\,
(-1)^L\rme^{\rmi \omega L/2}
\frac{D_{L}(-\rme^{\rmi\omega/2})}{D_{L}(-\rme^{-\rmi\omega/2})}\, H_1(\rme^{\rmi\omega/2})\hat\phi_G(\omega/2)\;.
\]
We now replace $\hat\phi_G$ by the expression obtained
in~\eqref{eqn:explicit-phi} and thanks to~\eqref{eqn:psi:filter},
\[
\hat\psi_G(\omega)=
(-1)^L\rme^{\rmi \omega L/2}
\frac{D_{L}(-\rme^{\rmi\omega/2})}{D_{L}(-\rme^{-\rmi\omega/2})}\rme^{-\rmi\omega/4} \rme^{\rmi\beta_L(\omega/2)}\hat\psi_H(\omega)\;.
\]
Since $D_{L}$ has a real impulse response and $-1=\rme^{\rmi\pi}=\rme^{-\rmi\pi}$, Lemma~\ref{lem:ratioD} gives that, for all
$\omega\in\R$,
$$
(-1)^L\rme^{\rmi \omega L/2}
\frac{D_{L}(-\rme^{\rmi\omega/2})}{D_{L}(-\rme^{-\rmi\omega/2})}
=\overline{\rme^{-\rmi L(\omega /2+\pi)}
\frac{D_{L}(\rme^{-\rmi(\omega/2+\pi)})}{D_{L}(\rme^{\rmi(\omega/2+\pi)})}}
= \rmi\; \rme^{\rmi\omega/4-\rmi \alpha_L(\omega/2+\pi)}\;.
$$
Hence, we finally get that, for all
$\omega\in\R$,
$$
\hat\psi_G(\omega)=\rmi\;\rme^{-\rmi \alpha_L(\omega/2+\pi)+\rmi\beta_L(\omega/2)}\hat\psi_H(\omega)\;.
$$
\eqref{eqn:explicit-psi} is proved.
\end{proof}

\subsection{Proof of Theorem~\ref{th:main-result-analyticity}}

\subsubsection*{Approximation of $1-\rme^{\rmi\eta_L}$}

We first state a simple result on the function $\rme^{\rmi\alpha_L}$.
\begin{lem}
  \label{lem:alpha}
  Let $L$ be a positive integer.
  The function $\alpha_L$ defined by~\eqref{eq:alpha-def} is $(4\pi)$-periodic.
  Moreover $\rme^{\rmi\alpha_L}$ is continous on $\R$ and we have, for all
  $\omega\in\R$, 
  \begin{equation}
    \label{eq:alpha-bound}
    \left|\rme^{\rmi\alpha_L(\omega)}-\mci(\omega)\right|\leq 2\sqrt2\;\Delta^{2L+1}(\omega)\;, 
  \end{equation}
  where
  \begin{equation}
    \label{eq:Komega}
    \mci(\omega)=
    \begin{cases}
      1&\text{ if $\omega\in[-\pi,\pi)+4\pi\Z$}\\
      -1&\text{ otherwise,}
    \end{cases}
  \end{equation}
  and
  \begin{equation}
    \label{eq:Delta-def}
    \Delta(\omega):=\min\left(|\tan(\omega/4)|,|\tan(\omega/4)|^{-1}\right)
  \end{equation}
\end{lem}
\begin{proof}
By definition~\eqref{eq:alpha-def}, $\alpha_L$ is $(4\pi)$-periodic and
continuous on $\R\setminus\left(2\pi+4\pi\Z\right)$. Moreover, at any of its discontinuity points in $2\pi+4\pi\Z$,
$\alpha_L$ jumps have height $2\pi$. Hence $\rme^{\rmi\alpha_L}$ is continous
over $\R$.

\noindent We now prove~\eqref{eq:alpha-bound}. We will in fact show the
following more precise bounds, valid for all $\omega\in\R$.
\begin{align}\label{ineq:cos:alpha:1}
|\cos(\alpha_L(\omega))-1|&\leq 2|\tan(\omega/4)|^{2(2L+1)}\;,\\
\label{ineq:cos:alpha:2}
|\cos(\alpha_L(\omega))+1|&\leq 2|\tan(\omega/4)|^{-2(2L+1)}\;,\\
\label{ineq:sin:alpha:1}
  |\sin(\alpha_L(\omega))|&\leq 2\Delta^{2L+1}(\omega)\;.
\end{align}
The bounds~\eqref{ineq:cos:alpha:1} and~\eqref{ineq:cos:alpha:2} easily follow from the identity
\[
\cos(\alpha_L(\omega))=\cos(2 \mathrm{arctan}(\tan(\omega/4)^{2L+1}))=\frac{1-\tan^{2(2L+1)}(\omega/4)}{1+\tan^{2(2L+1)}(\omega/4)}\;.
\]
The bound~\eqref{ineq:sin:alpha:1} follows from the identity
\[
\sin(\alpha_L(\omega))=\sin(2 \mathrm{arctan}(\tan^{2L+1}(\omega/4)))=\frac{2\tan^{2L+1}(\omega/4)}{1+\tan^{2(2L+1)}(\omega/4)}\;.
\]
The proof is concluded.
\end{proof}

Observe that by~\eqref{eq:beta-def} and the definition of $\eta_L$
in~\eqref{eq:eta-def}, $\rme^{\rmi \eta_L}$ can be
expressed directly from $\rme^{\rmi\alpha_L}$, namely as
$$ 
\rme^{\rmi \eta_L(\omega)}=\rme^{-\rmi\alpha_L(\omega/2+\pi)}\,\prod_{j\geq1}\rme^{\rmi\alpha_L\left(2^{-j-1}\omega\right)}\;.
$$
A quite natural question is to
determine the function $1-\rme^{\rmi \eta_L}$ obtained when
$\rme^{\rmi\alpha_L}$ is replaced by its large $L$ approximation $\mci$.  This
is done in the following result.
\begin{lem}
  \label{lem:approx-u-Linfty}
  Define the $(4\pi)$-periodic rectangular function $\mci$
  by~\eqref{eq:Komega}. Then, for all $\omega\in\R\setminus\{0\}$, we have
  \begin{equation}
    \label{eq:I-identity}
1-\mci(\omega/2+\pi)\,\prod_{j\geq1}\mci\left(2^{-j-1}\omega\right)=2\1_{\R_+}(\omega)
\;.    
  \end{equation}
\end{lem}
\begin{proof} 
  Note that the function $\omega\mapsto \mci(\omega+\pi)$ is the right-continuous,
  $(4\pi)$-periodic function that coincides with the sign of $\omega$ on
  $\omega\in[-2\pi,2\pi)\setminus\{0\}$. 
  It is then easy to verify that, by definition of $\mci$, we have, for all 
  $\omega\in\R$,
  \begin{align}
    \nonumber
    \mci(\omega)&=\mci(2\omega+\pi)\,\mci(\omega+\pi)\\
    \label{eq:telescope-prod-omega-neg}
             &=\frac{\mci(2\omega+\pi)}{\mci(\omega+\pi)}\\
       \label{eq:telescope-prod-omega-pos}
 &=\frac{-\mci(2\omega+\pi)}{-\mci(\omega+\pi)}\;.    
  \end{align}
(By periodicity of $\mci$, it only suffices to check the first equality on
$\omega\in[-2\pi,2\pi)$, the two other equalities follow, since $\mci$ takes values in $\{-1,1\}$.)  
Now, from the previous assertion, we have, for all $\omega<0$, that
$\mci(\omega 2^{-j}+\pi)=1$ for large enough $j$, and thus~\eqref{eq:telescope-prod-omega-neg} implies
$$
\prod_{j\geq1}\mci(2^{-j}\omega)
=\prod_{j\geq1}\frac{\mci(2^{-(j-1)}\omega+\pi)}{\mci(2^{-j}\omega+\pi)}
=\mci(\omega+\pi)\;,
$$
while, for all $\omega>0$, since $-\mci(\omega 2^{-j}+\pi)=1$ for large enough $j$,~\eqref{eq:telescope-prod-omega-pos} implies
$$
\prod_{j\geq1}\mci(2^{-j}\omega)
=\prod_{j\geq1}\frac{-\mci(2^{-(j-1)}\omega+\pi)}{-\mci(2^{-j}\omega+\pi)}
=-\mci(\omega+\pi)\;.
$$
Identity~\eqref{eq:I-identity} follows.
\end{proof}

We can now derive the main result of this section.
\begin{prop}
    \label{prop:uL-final-bound}
    Let $L$ be a positive integer.
  The function $\eta_L$ defined
  by~\eqref{eq:alpha-def},~\eqref{eq:beta-def} and~\eqref{eq:eta-def}
  satisfies the following bound, for all $\omega\in\R$,
  \begin{equation}
    \label{eq:uL-final-bound}
    \left|1-\rme^{\rmi\eta_L(\omega)}-2\1_{\R_+}(\omega)\right|
      \leq2\sqrt{2}\left(\Delta^{2L+1}(\omega/2+\pi)+\sum_{k=1}^{\infty}\Delta^{2L+1}(2^{-k-1}\omega)\right)
\;,
  \end{equation}
  where $\Delta$ is defined by~\eqref{eq:Delta-def}.
\end{prop}
\begin{proof}
  We have, for all $\omega\in\R$ and $J\geq1$,
  $$
  \prod_{j=1}^J  \rme^{\rmi \alpha_L(2^{-j}\omega)}-  \prod_{j=1}^J
  \mci(2^{-j}\omega)
  =\sum_{k=1}^J a_{k,J}(\omega) \;,
$$
where we denote
\[
a_{k,J}(\omega)=\left(\prod_{j=1}^{k-1}\rme^{\rmi
    \alpha_L(2^{-j}\omega)}\right)\cdot\left(\rme^{\rmi
      \alpha_L(2^{-k}\omega)}-\mci(2^{-k}\omega)\right)\cdot\left(\prod_{j=k+1}^{J}
  \mci(2^{-j}\omega)\right)\;,
\]
with the convention $\prod_1^0(\dots)=\prod_{J+1}^J(\dots)=1$. Since $\alpha_L$
is real valued and $\mci$ is valued in $\{-1,1\}$, it follows that, for all $\omega\in\R$ and $J\geq1$,
$$
\left|  \prod_{j=1}^J  \rme^{\rmi \alpha_L(2^{-j}\omega)}-  \prod_{j=1}^J
  \mci(2^{-j}\omega)
\right|\leq \sum_{k=1}^J |a_{k,J}(\omega)|\leq \sum_{k=1}^J
\left| \rme^{\rmi
      \alpha_L(2^{-k}\omega)}-\mci(2^{-k}\omega)
\right|\;.
$$
Applying Lemma~\ref{lem:alpha} yields for all $\omega\in\R$ and $J\geq1$,
$$
\left|  \prod_{j=1}^J  \rme^{\rmi \alpha_L(2^{-j}\omega)}- \prod_{j=1}^J
  \mci(2^{-j}\omega)
\right|\leq 2\sqrt{2} \left(\sum_{k=1}^{J}\Delta^{2L+1}(2^{-k}\omega)\right)\;.
$$
Letting $J\to\infty$ and applying the definition of $\beta_L$, we deduce that, for all $\omega\in\R$,
\begin{equation}
  \label{eq:beta-bound-Delta}
\left|\rme^{\rmi\beta_L(\omega)}-\prod_{j=1}^\infty
    \mci(2^{-j}\omega)\right|\leq 2\sqrt{2} \sum_{k=1}^{\infty}\Delta^{2L+1}(2^{-k}\omega)\;.  
\end{equation}
By definition of $\eta_L$, since $\alpha_L$ and $\beta_L$ are real valued and
$\mci$ is valued in $\{-1,1\}$, we have, for all $\omega\in\R$
\begin{align*}
  \left|\rme^{\rmi\eta_L(2\omega)}-\mci(\omega+\pi)\prod_{j=1}^\infty
  \mci(2^{-j}\omega)\right|
  &\leq
    \left|\rme^{\rmi\alpha_L(\omega+\pi)}-\mci(\omega+\pi)\right|
    +\left|\rme^{\rmi\beta_L(\omega)}-\prod_{j=1}^\infty
    \mci(2^{-j}\omega)\right|\;.
\end{align*}
Hence, with Lemma~\ref{lem:alpha} and~\eqref{eq:beta-bound-Delta}, we conclude,
for all $\omega\in\R$,
$$
  \left|\rme^{\rmi\eta_L(2\omega)}-\mci(\omega+\pi)\prod_{j=1}^\infty
    \mci(2^{-j}\omega)\right|
  \leq2\sqrt{2}\left(\Delta^{2L+1}(\omega+\pi)+\sum_{k=1}^{\infty}\Delta^{2L+1}(2^{-k}\omega)\right)
\;.
$$
The bound~\eqref{eq:uL-final-bound} then follows from
Lemma~\ref{lem:approx-u-Linfty}.
\end{proof}

\subsubsection*{Study of the upper bound}

The objective is to simplify the right-hand side of~\eqref{eq:uL-final-bound} to obtain the form given in~\eqref{eq:final-analycity-result}. The following lemma essentially gives some interesting properties of the
function $\Delta$. 
\begin{lem}\label{lem:delta} 
  Let $\Delta$ be the function defined by~\eqref{eq:Delta-def}. Then $\Delta$ is an even
  $(2\pi)$-periodic function, increasing and bijective from $[0,\pi]$ to
  $[0,1]$. It follows that, for all $\omega\in\R$,
  \begin{equation}
    \label{eq:Delta1bound}
  \Delta(\omega)=\tan\left(\frac{\pi}{4}(1-\frac{1}{\pi}\delta(\omega,\pi+2\pi\Z))\right)\leq
  1-\frac1\pi \delta(\omega,\pi+2\pi\Z)\;,    
  \end{equation}
where $\delta$ is the function defined in equation~\eqref{eqn:delta}.
\end{lem}
\begin{proof}
  The proof is straightforward and thus omitted. 
\end{proof}
  Note that the upper bound in~\eqref{eq:Delta1bound} decreases from 1 to 0 as
$\delta(\omega,\pi+2\pi\Z)$ increases from 0 to $\pi$.
Since $\Delta$ takes its values in $[0,1)$ on
$\R\setminus\left(\pi+2\pi\Z\right)$, Lemma~\ref{lem:alpha} shows that, out of the set
$\pi+2\pi\Z$, $\rme^{\rmi\alpha_L}$ uniformly converges to the
$(4\pi)$-periodic rectangular function $\mci$ as $L\to\infty$.

We will use the following bound.
\begin{lem}\label{ineq:key}
For all $\omega\in (-\pi/4,\pi/4)$ and $L\geq0$, we have
\[
\sum_{j=0}^\infty |\tan|^{2L+1}(2^{-j}\omega)\leq 2\left|\tan\right|^{2L+1}(\omega)\;.
\]
\end{lem}
\begin{proof}
  It suffices to prove the inequality for $\omega\in(0,\pi/4)$.
  By convexity of $\tan$, the slope $x\mapsto x^{-1}\tan(x)$ is
  increasing on $[0,\pi/4)$ and so is $x\mapsto x^{-1}\tan^{2L+1}(x)$ for
  $L\geq0$. Hence we have, for all $\omega\in(0,\pi/4)$,
  \begin{align*}
    \sum_{j=0}^\infty \tan^{2L+1}(2^{-j}\omega)
    &=\sum_{j=0}^\infty
      2^{-j}\omega\left(2^{-j}\omega\right)^{-1}\tan^{2L+1}(2^{-j}\omega)    \\
    &\leq\sum_{j=0}^\infty
      2^{-j}\omega\left(\omega\right)^{-1}\tan^{2L+1}(\omega)\\
    &=2\,\tan^{2L+1}(\omega)\;.
  \end{align*}
    The proof is concluded.
\end{proof}

We also have the following lemma.
\begin{lem}
  For all $\omega\in\R$, we have 
  \begin{align}
    \label{eq:1}
    \delta(\omega/2+\pi,\pi+2\pi\Z)&\geq 2^{-1}\delta(\omega,4\pi\Z)\;,\\
    \delta(2^{-j}\omega,\pi+2\pi\Z)&\geq 2^{-j}\,\delta(\omega,4\pi\Z)\qquad\text{for all
                                integer $j\geq2$}\;.
    \label{eq:2}
  \end{align}
\end{lem}
\begin{proof}
  The bound~\eqref{eq:1} is obvious. To show~\eqref{eq:2}, take
  $x\in\pi+2\pi\Z$. Then, for all $\omega\in\R$ and $j\geq2$, we have
  $\left|2^{-j}\omega-x\right|=2^{-j}\left|\omega-2^jx\right|$ and, since $2^j
  x\in4\pi\Z$, we get~\eqref{eq:2}.  
\end{proof}

We are now able to give a more concise upper bound.
\begin{lem}
  \label{lem:sumDeltaBound}
  Let $\Delta$ be defined by~\eqref{eq:Delta-def}. Then, for all
  $\omega\in\R$, we have
  \begin{equation}
  \label{eq:sumDeltaBound}
\sum_{j=1}^{\infty}\Delta^{2L+1}(2^{-j-1}\omega)
\leq\left(\log_2\left(\frac{\max(4\pi,|\omega|)}{2\pi}\right)+1\right)\,\left(1-\frac{\delta(\omega,4\pi\Z)}{\max(4\pi,|\omega|)}\right)^{2L+1}\;.
\end{equation}
\end{lem}
\begin{proof}
  Denote
$$
\iota(\omega)=\min\left\{j\geq1~:~|\omega|\,2^{-j-1}<\pi\right\}\leq
\log_2\left(\frac{\max\left(4\pi,|\omega|\right)}{2\pi}\right)\;.
$$
$$
\sum_{j=1}^{\infty}\Delta^{2L+1}(2^{-j-1}\omega)
\leq\sum_{j=1}^{\iota(\omega)-1} \Delta^{2L+1}(2^{-j-1}\omega)+ \sum_{j\geq \iota(\omega)}|\tan|^{2L+1}(2^{-j-3}\omega)
$$
Lemma~\ref{ineq:key} gives that, for all
$\omega\in\R$,
$$
\sum_{j\geq \iota(\omega)}|\tan|^{2L+1}(2^{-j-3}\omega)\leq 2|\tan|^{2L+1}(2^{-\iota(\omega)-3}\omega)=2\Delta^{2L+1}(2^{-\iota(\omega)-1}\omega)\;.
$$
The last two bounds yield, for all $\omega\in\R$,
$$
\sum_{j=1}^{\infty}\Delta^{2L+1}(2^{-j-1}\omega)
\leq \left(\iota(\omega)+1\right)\left(\sup_{1\leq j\leq \iota(\omega)}\Delta(2^{-j-1}\omega)\right)^{2L+1}\;.
$$
Note that Lemma~\ref{lem:delta} and~\eqref{eq:2} imply
$$
\sup_{1\leq j\leq \iota(\omega)}\Delta(2^{-j-1}\omega)\leq
1-\frac1\pi2^{-\iota(\omega)-1}\,\delta(\omega,4\pi\Z)\;.
$$
The above bound on $\iota(\omega)$ then
gives~\eqref{eq:sumDeltaBound}.
\end{proof}
We can now conclude with the proof of the main result. 
\begin{proof}[Proof of Theorem~\ref{th:main-result-analyticity}]
By Lemma~\ref{lem:delta} and~\eqref{eq:1}, we have, for all $\omega\in\R$,
$$
\Delta(\omega/2+\pi)\leq1-\frac1{2\pi} \delta(\omega,4\pi\Z)\leq1-\frac{\delta(\omega,4\pi\Z)}{\max(4\pi,|\omega|)}\;.
$$
Using this bound,~\eqref{eq:uL}, Proposition~\ref{prop:uL-final-bound} and
Lemma~\ref{lem:sumDeltaBound}, we get~\eqref{eq:final-analycity-result}.
\end{proof}

\section{Proofs of Section~\ref{prop:R}}
\label{sec:proofs-perfect-reconstruction}
The following lemma will be useful.
\begin{lem}
  \label{lem:zerosPS}
  Let $L$ be a positive integer.  The complex roots of the polynomial
  $\tilde s(x)=\sum_{n=0}^L{{2L+1}\choose{2n}}x^n$ belong to $\R_-$.
\end{lem}
\begin{proof}
  Observe that for all $z\in\C$, $\tilde s(z^2)=\frac12((1+z)^{2L+1}+(1-z)^{2L+1})$.
  Thus if $\tilde s(z^2)=0$ with $z=x+\rmi y$ and $(x,y)\in\R^2$, we necessarily have that
$|1+z|^2=(1+x)^2+y^2$ is equal to $|1-z|^2=(1-x)^2+y^2$, and thus $x=0$ and $z^2\in\R_-$.
\end{proof}
\begin{proof}[Proof of Proposition~\ref{prop:R}]
By Proposition~\ref{prop:D}, we have, for all  $\theta\in\R$,
\begin{align}\nonumber
D_{L}(\rme^{2\rmi\theta})D_{L}(\rme^{-2\rmi\theta})&=\frac{1}{4(2L+1)^2}\left|(1+\rme^{\rmi\theta})^{2L+1}
+(1-\rme^{\rmi\theta})^{2L+1}\right|^2\nonumber\\
\nonumber
&=\frac{2^{2(2L+1)}}{4(2L+1)^2} \left[\cos^{2(2L+1)}(\theta/2)
+\sin^{2(2L+1)}(\theta/2)\right]\\
\nonumber
&=\frac{2^{2L}}{(2L+1)^2}\left[\frac{(1+\cos(\theta))^{2L+1}
  +(1-\cos(\theta))^{2L+1}}{2}\right]\\
\label{eq:DDeq}  
&=\frac{2^{2L}}{(2L+1)^2}\sum_{n=0}^L{{2L+1}\choose{2n}}\cos^{2n}(\theta)\;.
\end{align}
Note that if
$z=\rme^{2\rmi\theta}$ with $\theta\in\R$, then
$2+z+1/z=2(1+\cos(2\theta))=(2\cos(\theta))^2$.
By definition of $S$, we obtain that, for all $\theta\in\R$, 
\begin{align*}
S(\rme^{2\rmi\theta})&=2^{2M} \cos^{2M}(\theta)\times \frac{2^{2L}}{(2L+1)^{2}}\sum_{n=0}^L{{2L+1}\choose{2n}}\cos^{2n}(\theta) \\
&=\frac{2^{2M+2L}}{(2L+1)^2}\,s(\cos^2(\theta))\;,
\end{align*}
where $s$ is
the polynomial defined by~\eqref{eq:PS-def}.
Looking for a solution $R$ of~\eqref{eqn:PR} in the form
$R(z)=U\left(\frac{2+z+1/z}{4}\right)$ with $U$ real polynomial and focusing on
$z=\rme^{2\rmi\theta}$ with $\theta\in\R$, we obtain the equation
\begin{equation}
\label{eqn:eqnR}
U(1-y) s(1-y)+U(y) s(y)=C(L,M)^2.
\end{equation}
where we have denoted $y=\sin^2(\theta)=1-(2+z+1/z)/4\in[0,1]$ and $C(L,M)={(2L+1)2^{-M-L+1/2}}$.  Reciprocally, any such polynomial $U$ provides a solution
$R(z)=U\left(\frac{2+z+1/z}{4}\right)$ of~\eqref{eqn:PR} for
$z=\rme^{2\rmi\theta}$ with $\theta\in\R$ and then for all $z\in\C^*$ by
analytic extension.

Since the
complex roots of $s$ are valued in the set $\R_-$ of non-positive real numbers
(see Lemma~\ref{lem:zerosPS}), we get that $s(1-y)$ and $s(y)$ are prime
polynomials of degree $L+M$. Thus the Bezout Theorem allows us to describe the couples of real
polynomials $(U,V)$ solutions of the equation 
$$
V(y)s(1-y)+U(y)s(y)=C(L,M)^2\;.
$$
Note that $U$ is a solution of~\eqref{eqn:eqnR} if and only
if $(U,V)$ is a solution of the Bezout equation with $V(y)=U(1-y)$. Now, by
uniqueness of the solution of the Bezout equation such that both $U$ and $V$
have degrees at most $L+M-1$, we see that this solution must satisfy
$V(y)=U(1-y)$ (since otherwise $(V(1-y),U(1-y))$ would provide a different
solution). Hence we obtain a unique solution $U=r$ of~\eqref{eqn:eqnR} of
degree at most $L+M-1$, which proves Assertion~\ref{item:1}.

Other solutions $(U,V)$ of the Bezout equation are obtained by taking
$U(y)=r(y)+s(1-y)\,q(y)$ with $q$ any polynomial. Looking for such a solution
of~\eqref{eqn:eqnR}, we easily get that it is one if and only if $q$ satisfies
$q(1-y)=-q(y)$. The proof of Assertion~\ref{item:2} is concluded.
\end{proof}
\begin{proof}[Proof of Theorem~\ref{thm:Q}]
  By Proposition~\ref{prop:R}, the given $R$ is a solution of~\eqref{eqn:PR}
  provided that $q$ satisfies $q(1-y)=-q(y)$, which we assume in the
  following. As explained above, the
  factorization holds if and only if $R$ is non-negative on the unit circle,
  or, equivalently, if $r(y)+s(1-y)\,q(y)$ is non-negative for $y$ in $(0,1)$. By antisymmetry of
  $q$ around $1/2$, this is equivalent to have, for all $y\in[1/2,1)$,
$$
r(y)+s(1-y)\,q(y)\quad\text{and}\quad r(1-y)-s(y)\,q(y)\geq0\;.
$$
Since $s(y)>0$ for all  $y\in(0,1)$ and using that~\eqref{eqn:eqnR} holds
with $U=r$, we finally obtain that the claimed
factorization holds if and only if, for all $y\in[1/2,1)$,
\begin{equation}
  \label{eq:F-2ndCond}
-\frac{r(y)}{s(1-y)}\leq q(y)\leq
-\frac{r(y)}{s(1-y)}+\frac{C(L,M)^2}{s(y)s(1-y)} \;.  
\end{equation}
We first note that for all $y\in[1/2,1)$,
$1/(s(y)s(1-y))\geq1/(s(1)s(1/2))>0$. Hence the upper bound condition
in~\eqref{eq:F-2ndCond} is away from the lower bound by at least a positive
constant over $y\in[0,1/2)$.  Second, using again~\eqref{eqn:eqnR} with $U=r$
at $y=1/2$ we have
$$
2r(1/2)s(1/2)=C(L,M)^2\;.
$$
It follows that, for $y=1/2$,~\eqref{eq:F-2ndCond} reads 
$$
-\frac{C(L,M)^2}{2s^2(1/2)}\leq q(1/2)\leq\frac{C(L,M)^2}{2s^2(1/2)} \;,  
$$
This is compatible with $q(1/2)=0$ inherited by the antisymmetric property of $q$
around $1/2$. We conclude by applying the Stone-Weierstrass theorem to obtain
the existence of a real polynomial $q$ satisfying~\eqref{eq:F-2ndCond} for all
$y\in[1/2,1)$, $q(y)=-q(1-y)$ for all $y\in\R$.
\end{proof}

\section{Proofs of Section~\ref{sec:simul}}
\label{sec:defin-polyn-r}

We start with a result more precise than Lemma~\ref{lem:zerosPS}.
\begin{lem}\label{lem:roots}
  Let $L$ be a positive integer.  The complex roots of the polynomial
  $\tilde s(y)=\sum_{n=0}^L{{2L+1}\choose{2n}}y^n$  are the
  $y_{0,L},\dots,y_{L-1,L}$ defined in~(\ref{eq:s-roots}).
\end{lem}
\begin{proof}
  Recall that  that 
  $\tilde s(z^2)=\frac{1}{2}\left[(1+z)^{2L+1}+(1-z)^{2L+1}\right]$ and that
  $\tilde s(z^2)=0$ is equivalent to
$z\neq 1$ and $\left(\frac{1+z}{1-z}\right)^{2L+1}=-1$,   
that is
\[
\frac{1+z}{1-z}=\rme^{\rmi (\pi+2k\pi)/(2L+1)}\;,\qquad k\in \{-L,\cdots,L\}\;.
\]
There is no such $z$ for $k=L$ and for $k\in
\{-L,\cdots,L-1\}$, this is the same as
\[
z=\frac{\rme^{\rmi (\pi+2k\pi)/(2L+1)}-1}{1+\rme^{\rmi
    (\pi+2k\pi)/(2L+1)}}=\rmi \,\tan\left(\frac{\pi(2k+1)}{2(2L+1)}\right)\;.
\]
Taking the square and keeping only $k=0,\dots,L$ to get distinct roots we get
the result.  
\end{proof}
\begin{proof}[Proof of Proposition~\ref{pro:expr-rL0}]
  Using Lemma~ \ref{lem:roots} and the Bezout equation $[B(L,0)]$, we have that
  for any $k\in\{0,\cdots,L-1\}$
\[
r_{L,0}(1-y_{k,L}) =\frac{(2L+1)^2\,2^{-2L+1}}{s_{L,0}(1-y_{k,L})}\;.
\]
Since the polynomial $r_{L,0}$ has degree at most $L-1$ with known values
in the $L$ distinct points $1-y_{k,L}$, we deduce its explicit expression given in
Proposition~\ref{pro:expr-rL0} by a standard interpolation formula. 

We conclude with the proof of the recursive relation~(\ref{eq:recursiv-formula-r}).
Since for any $M\geq 1$, $s_{L,M}(y)=y s_{L,M-1}$, the polynomial $r_{L,M}$ satisfies
\begin{equation}\label{eqn:PR-MGene}
r_{L,M}(1-y) \times [(1-y) s_{L,M-1}(1-y)]+r_{L,M}(y) \times [y s_{L,M-1}(y)]=(2L+1)^2\,2^{-2M-2L+1}\;.
\end{equation}
We deduce that $2^{2} y r_{L,M}(y)$ satisfies equation $[B(L,M-1)]$ and by
Proposition~\ref{prop:R}, 
\begin{equation}\label{eqn:rec}
2^{2} y r_{L,M}(y)=r_{L,M-1}(y)+s_{L,M-1}(1-y)\,q(y)\;,
\end{equation}
where $q$ is a polynomial satisfying $q(1-y)=-q(y)$.
Since the degree of $r_{L,M}$ is at most $L+M-1$ and that of $s_{L,M-1}$ is
$L+M-1$, we get that $q$ takes the form
$q(y)=q(0)(1-2y)$ and it only remains to determine $q(0)$.
Note that $s_{L,M}(1)=2^{2L}$, so~(\ref{eqn:rec}) yields
$q(0)=-2^{-2L}r_{L,M-1}(0)$ and we finally
obtain~(\ref{eq:recursiv-formula-r}).
\end{proof}

\bibliographystyle{plainnat}
\bibliography{cwt}
\end{document}